\documentclass[12pt,draftclsnofoot,onecolumn]{IEEEtran}
\usepackage[cmex10]{amsmath}
\usepackage{enumerate}
\usepackage{algorithmicx}
\usepackage[ruled]{algorithm}
\usepackage{algpseudocode}
\usepackage{algpascal}
\usepackage{algc}
\usepackage{amsthm}
\usepackage{amssymb}
\usepackage{makeidx}
\usepackage{graphicx}          
\usepackage[dvips]{epsfig}    

\begin{document}
\title{Distributed Nonlinear MPC of Multi-Agent Systems with Data Compression and Random Delays - Extended Version}
\author{Sami~El-Ferik,~Bilal~A.~Siddiqui~and~Frank~L.~Lewis
\thanks{S. El-Ferik and B. Siddiqui are with Systems Engineering Department, King Fahd University of Petroleum \& Minerals, KSA. \{selferik,airbilal\}@kfupm.edu.sa.
F.L. Lewis is with Research Institute, University of Texas at Arlington, Texas. lewis@uta.edu}
%\thanks{Manuscript received September xx, 2012; revised January xx, 2013.}
}
\markboth{Accepted for IEEE Transactions on Automatic Control.}
{El-Ferik and Siddiqui: Collision Free Distributed NMPC}
%\IEEEpubid{0000--0000/00\$00.00~\copyright~2013 IEEE}
\maketitle
%\thanks[footnoteinfo]{The author(s) would like to acknowledge the support provided by King Abdulaziz City for Science
%and Technology (KACST) through the Science and Technology Unit at King Fahd University of
%Petroleum and Minerals (KFUPM) for funding this work through project No. 09-SPA783-04  as part of the National Science, Technology and Innovation Plan}
\begin{abstract}                          % Abstract of not more than 250 words.
This is an extended version of a technical note accepted for publication in IEEE Transactions on Automatic Control. The note proposes an Input to State practically Stable (ISpS) formulation of distributed nonlinear model predictive controller (NMPC) for formation control of constrained autonomous vehicles in presence of communication  bandwidth limitation and transmission delays. Planned trajectories are compressed using neural networks resulting in considerable reduction of data packet size, while being robust to propagation delays and uncertainty in neighbors' trajectories. Collision avoidance is achieved by means of spatially filtered potential field. Analytical results proving ISpS and generalized small gain conditions are presented for both strongly- and weakly-connected networks, and illustrated by simulations.
\end{abstract}

\section{Introduction}\label{S:intro}
Cooperation between autonomous vehicles has shown promising advantages in terms of robustness, adaptivity, reconfigurability, and scalability. A prevalent technique for formation control is MPC for its inherent ability to handle constraints and uncertainty. Dunbar et al \cite{Dunbar2012} considered distributed NMPC for synchronization of agents  by broadcasting state error trajectories  to the immediate neighbors. A generalized framework for distributed NMPC for cooperative control is proposed in \cite{Allgower2012}, where asymptotic stability is ensured by terminal constraint set. A framework for quasi-parallel NMPC without restriction of terminal set, extended to the multi-agent case recently is shown to be asymptotically stable \cite{Pannek2013}. Distributed NMPC was considered for a group of strongly connected agents receiving delayed input from their neighbors in \cite{Franco2008}-\cite{Chao2012}. The delayed information is projected in the prediction horizon using either a \textit{time-based forward forgetting-factor} or by \textit{linear recurrence}, respectively. Collision avoidance (CA) within MPC framework is well studied for linear systems \cite{Casavola2014}, but similar work in nonlinear MPC setting is still rare. CA  among multiple vehicles is achieved by adding a repelling potential field to local NMPC cost function and transmitting the entire planned trajectory \cite{Yoon2007}. Priority strategy for CA in NMPC framework, using neighbors' randomly delayed information has been proposed in \cite{Chao2012}. Hierarchical multi-level control is considered in \cite{Chaloulos2010} by combining potential field with linear MPC, such that only the first step of the trajectory is optimized and linear recursion is used to predict the trajectory over the remaining horizon. Stability proofs are unavailable in most of these CA works. In this note, we address fleet control with collision avoidance of constrained autonomous vehicles subject to limited network throughput and propagation delays by employing distributed NMPC control. Each agent performs local optimization based on an estimate of planned trajectories received from neighboring agents. Since network throughput is assumed limited,  exchanged trajectories are compressed using neural networks (NN) as a universal approximator. This property is crucial in our stability analysis, since the impact of estimation error on system dynamics is considered as a bounded non-vanishing (persistent) disturbance. Correction for propagation delays is achieved by time-stamping each communication packet \cite{Srinavas2004}. Collision avoidance is achieved by formulating a new spatially-filtered repelling potential field which is activated in a ``gain-scheduling" type of approach to avoid transforming the problem into  mixed-integer nonlinear programming. We prove this distributed control strategy to be ISpS for heterogeneous agents connected in strongly- or weakly-connected network, robust to uncertainty in neighbors' planned trajectories. This algorithm is an improvement over \cite{Franco2008} and contributes to the literature with the following original results: (a) Only an approximation of planned trajectories is transmitted; (b) NN-based data compression algorithm is used in compressing the planned trajectories; (c) collision avoidance  by using a spatially filter potential function with rigorous stability proofs; (d) new ISpS and generalized small gain conditions are derived to ensure stability of  proposed  algorithm; (e) stability results are extended even to weakly connected networks.
\section{Preliminaries}
Let $L_2$ Euclidean norm be denoted by $|\cdot|$ and let $|\cdot|_{\infty}$ be the $ L_\infty$ norm. The identity function is denoted by $\mathcal{I}: \mathbb{R}\to\mathbb{R}$, functional composition of two functions $\gamma_1$ and $\gamma_2$ by $\gamma_1\circ\gamma_2$ and function inverse of function $\alpha$ by $\alpha^{-1}$. For a set $A\subseteq\mathbb{R}^n$, the point to set distance from $\zeta\in\mathbb{R}$ to $A$ is denoted by $d\left(\zeta,A\right)\triangleq\inf\left\{\mid\eta-\zeta\mid,\eta\in{A}\right\}$. The difference between two sets $A,B\subseteq\mathbb{R}^n$ is denoted by $A\backslash{B}\triangleq\left\{x:x\in{A},x\notin{B}\right\}$. An indicator function of vector $x$  defined as $\mathbf{1}_{x>0} = \{1\, \text{if } x \succ 0,\,0\, \text{otherwise}\}$, where $\succ$ is element-wise inequality. We also use class $\mathcal{K, K_\infty}$ and $\mathcal{KL}$ comparison functions \cite{Sontag1996ISS}. Consider the discrete-time nonlinear system $x_{t+1}=f(x_t,w_t)$ with $f\left(0,0\right)={0}$, where $x_t\in\mathbb{R}^n$ and $w_t\in\mathbb{R}^r$ are state and external input respectively. If $x_t\in\Xi, \forall t > t_0$ whenever $x_{t_0}\in\Xi$ and bounded input $w_t\in{W}$, then $\Xi$ is called a Robust Positively Invariant (RPI) set. Moreover, if $\Xi$ is compact, RPI and contains the origin as an interior point, the system $x_{t+1}=f(x_t,w_t)$ is said to be regionally Input-to-State practically Stable (ISpS) in $\Xi$  for $x_{0}\in{\Xi}$ and $w\in{W}$, if there exists $\mathcal{KL}$-function $\beta$, $\mathcal{K}$-function $\gamma$ and constant $c>0$ such that
\begin{eqnarray} \label{E:ISpS_def}
\left| {x_t} \right| \le \beta \left( {\left| {{x_0}} \right|,t} \right) + \gamma \left( {\left | w \right |_{\infty}} \right)+c
        \end{eqnarray}
If $c\equiv 0$ , then the system is said to be regionally Input-to-State Stable (ISS) in $\Xi$ \cite{Sontag1996ISS}.  Function $V: \mathbb{R}^n\times\mathbb{R}^n\to\mathbb{R}_{\ge{0}}$ is an ISpS Lyapunov function in $\Xi$, if for suitable functions $\alpha_{1,2,3},\sigma_{3}\in\cal{K}_{\infty}$,  $\sigma_{1,2}\in\cal{K}$ and constants $\bar{c},\bar{\bar{c}}> 0$, there exists a compact and RPI set $\Xi$ and another set $\Omega\subset{\Xi}$ with origin as an interior-point ($\Omega$ is also RPI), such that the following conditions hold,
\begin{equation}\label{E:lyapISS_cond1}
V( {x_t,w_t} ) \ge {\alpha _1}( {| x_t|} ),\,\,\,\,\,\,\,\,\,\,\forall \,x_t \in \Xi
\end{equation}
\begin{equation}\label{E:lyapISS_cond2}
\begin{array}{l}
V\left( {f\left( {{x_t},{w_t}} \right),{w_{t + 1}}} \right) - V\left( {{x_t},{w_t}} \right) \le \\
 - {\alpha _2}\left( {\left| {{x_t}} \right|} \right) + {\sigma _1}\left( {\left| {{w_t}} \right|} \right) + {\sigma _2}\left( {\left| {{w_{t + 1}}} \right|} \right) + \bar c,{\kern 1pt} {\kern 1pt} {\kern 1pt} {\kern 1pt} {\kern 1pt} \forall {\kern 1pt} {x_t} \in \Xi 
\end{array}
\end{equation}
\begin{equation}\label{E:lyapISS_cond3}
V\left( {x_t,w_t} \right) \le {\alpha _3}\left( {\left| x_t \right|} \right) + {\sigma _3}\left( {\left| w_t \right|} \right)+\bar{\bar{c}},\,\,\,\,\,\,\,\,\,\,
 \forall \,x_t \in \Omega
\end{equation}
The relation between ISpS Lyapunov functions and ISpS is shown in Theorem \ref{T:ISpS_gen}. ISS implies ISpS, but  converse is not true, since an ISS system with $0-$input, i.e. $w_k = 0, \forall  k\ge 0$  implies asymptotic stability to the origin, while for an ISpS system, $0-$input implies asymptotic stability to a compact set (ball of radius $c$) containing the origin. In this paper, the stability analysis will demonstrate that according to the proposed control approach, closed-loop dynamics is ISpS, not ISS, due to uncertainty resulting from data compression. Thus, in this study, $c$ in equation (\ref{E:ISpS_def})  is not zero but function of bounded error in NN estimation. Information exchange among networked vehicles is conveniently modeled by general mixed graphs (directed and undirected edges). An information graph is a set of nodes $A^i$ and edges connecting node pairs $E(A^i,A^j)$. Define connectivity matrix as $\Gamma=[\bar\gamma_{ij}]$, where $\bar\gamma_{ij}>0$ if $(A^i,A^j)\in E$ and $0$ otherwise (by convention $\bar\gamma_{ii}=0$). Neighborhood of a node $A^i$ is $G^i:= \{A^j:\bar\gamma_{ij}>0\}\cup\{A^j:\bar\gamma_{ji}>0\}$. A network is said to be \emph{strongly connected} if there is an undirected path from any node to any other node in the network. In this case, connectivity gain matrix $\Gamma$ is irreducible. A network is said to be weakly connected if there are at least two nodes for which a directed path connecting them does not exist. For weakly connected networks, connectivity gain matrix $\Gamma$ can be reduced to upper block triangular form \cite{Dashkovskiy2010}. Next, we will formulate the distributed multi-agent problem.
\subsection{Distributed Multi-Agent NMPC with Collision Avoidance}\label{Sec:DNMPC-CA}
Consider a set of $N$ agents $A^i$ each having nonlinear discrete-time dynamics:
\begin{eqnarray} \label{E:agentsys}
	x^i_{t+1}=f^i(x^i_t,u^i_t),\,\,\,\,\,\,\,\,\forall {t}\ge0,\,\,\ i=1,\dots,{N}
\end{eqnarray}
Local states $x^i_t$ and control inputs $u^i_t$  belong to constrained sets $x_t^i \in {X^i} \subset {\mathbb{R}^{{n^i}}},{\mkern 1mu} {\mkern 1mu} {\mkern 1mu} {\mkern 1mu} {\mkern 1mu} {\mkern 1mu} {\mkern 1mu} {\mkern 1mu} \;u_t^i \in {U^i} \subset {\mathbb{R}^{{m^i}}}$. 
Agents are decoupled from each other in open loop. On the other hand, closed-loop  control takes into account the neighbors' states and therefore couples the dynamics. Let $\tilde{w}^i_t$ be the approximation of trajectories $w^i_t=\{x^j_t\}, \forall j\in{G^i}$ of neighbors of $A^i$, such that $w_t^i \in {W^i} \subset {\mathbb{R}^{{p^i}}}$. For each agent $A^i$, the general finite-horizon cost function is defined as:
\begin{equation} \label{E:cost1}
\begin{array}{l}
J_t^i = \sum\limits_{l = t}^{t + N_p^i - 1} {\left[ {{h^i}\left( {x_l^i,u_l^i} \right) + {q^i}\left( {x_l^i,\tilde w_l^i} \right)} \right]}+h_f^i({x_{t + N_p^i}^i})
\end{array}
\end{equation}
where $N_p^i$ and $N_c^i$ are prediction and control horizons respectively. Distributed cost \eqref{E:cost1} consists of local transition cost $h_l^i$, local terminal cost $h_f^i$ and interaction cost $q^i_l$, see \cite{Franco2008} for details. We define an agent $A^i$ to be on collision course with at least one other agent if $\sum\limits_{j \in {G^i}} {{{\bf{1}}_{(R_{\min}^i - d^{ij}_k) > 0,\forall t \le k \le (t + N_p^i)}}} >0$, where $R_{\min}$ is the safety zone of an agent and $d^{ij}_k$ is the Euclidean distance between agent $A^i$ and $A^j$. Repelling potential can be formulated as:
\begin{eqnarray} \label{E:colav_pot}
\Phi _t^i = \sum\limits_{j \in {G^i}} {\frac{{\bar \lambda R_{\min}^i{{ \mathbf{1}}_{(R_{\min}^j - d^{ij}_k) > 0,\forall t \le k \le (t + N_p^i)}}}}{{\sum\limits_{k = t}^{t + N_P^i} {\lambda (d^{ij}_k) d^{ij}_k} }}} {\mkern 1mu} {\mkern 1mu}
\end{eqnarray}
where $0<\lambda_{\min}\le\lambda(d_{ij})\le\lambda_{\max}$ are positive weights of a  filter and are strictly decreasing in their argument, such that $\bar{ \lambda}\triangleq\sum\limits_{k = t}^{t + N_P^i} {\lambda (d^{ij}_k)} $.  If at any instant $t \le k \le (t + N_p^i)$ in the prediction horizon, an agent $A^i$ has a feasible trajectory which falls within $R_{\min}^j$ of agent $A^j$, the repelling potential \eqref{E:colav_pot} becomes non-zero. To cater for collision course, cost \eqref{E:cost1} is modified as
\begin{equation}\label{E:colav_cost}
\acute J_t^i = J_t^i(1 + \Phi _t^i)
\end{equation}
Strength of potential field \eqref{E:colav_pot} is inversely proportional to the weighted average distance between the two agents ${\bar{d}^{ij}_t=\sum\limits_{k = t}^{t + N_P^i} {\lambda (d^{ij}_k)d^{ij}_k} }/\bar{\lambda}$. The weights $\lambda$, strictly decreasing with $d^{ij}_k$, ensure that the smallest separation between two agents gets the highest weight. On the other hand, taking a simple average (i.e. $\lambda\equiv 1$) or a time-based forgetting factor ($\lambda$ is strictly decreasing with $k$, the time index), results in poor performance in collision avoidance, as trajectories which enter very late in zone $R_{\min}$ (i.e. $R_{\min}^i - d^{ij}_k > 0, k \to t + N_p^i$) have a small repelling potential \eqref{E:colav_pot}, and hence not prevented from very early on. Such strategy results in agents getting very close before they start repelling each other to avoid collision. However, with cost \eqref{E:colav_cost}, trajectories  are immediately penalized upon falling within zone $R^j_{\min}$ and are obviously avoided in the NMPC optimization. The indicator function in \eqref{E:colav_pot} acts as a ``gain-scheduled" binary (0-1) variable depending on whether a feasible trajectory falls within $R_{\min}$. We define successful collision avoidance to occur if weighted average distance between the agents on collision course increases i.e.
\begin{equation}\label{E:colav_exp_cond}
\sum\limits_{k = t}^{t + N_P^i} {\lambda (d^{ij}_k)d^{ij}_k}  < \sum\limits_{k = t + 1}^{t + N_P^i + 1} {\lambda (d^{ij}_k)d^{ij}_k} 
\end{equation}
Control sequence $u^i_{t,t+N_p^i}$ consists of  $u^i_{t,t+N_c^i-1}$ and $u^i_{t+N_c^i,t+N_p^i-1}$. The latter part is generated by local \emph{auxiliary} control law $u_i^l=k^i_f(x_l^i)$ for ${t\ge{N_c^i}}$, while the former is the distributed optimal control $u^i_{t,t+N_c^i}$ which is the solution of  the Problem \ref{P:RHOCP}. Suboptimal $u^i_{t,t+N_c^i-1}$ satisfying all constraints is called \emph{feasible} control.
\newtheorem{pbm}{Problem}
\begin{pbm}\label{P:RHOCP}
At every instant $t\ge{0}$ for each agent, given horizons $N_p^i$ and $N_c^i$, and auxiliary control $k_f^i$, find the optimal control sequence ${u}^{i,\star}_{t,t+N_c^i-1}$, which minimizes distributed finite horizon cost \eqref{E:cost1} (or \eqref{E:colav_cost} for collision avoidance), satisfies state and input constraints and system dynamics \eqref{E:agentsys}, such that the terminal state is constrained to a terminal set, i.e. $x^i_{t+N_p^i}\in{X^i_f}$. In the receding horizon strategy, only the first element of ${u}^{i,\star}_{t,t+N_c^i-1}$ is implemented at each instant, such that the closed loop dynamics becomes
\begin{equation}\label{E:sys_cl}
	x^i_{t+1}=f^i(x^i_t,u^{i,\star}(x^i_t,w^i_t))=\tilde{f}^i(x^i_t,w^i_t)
\end{equation}
\end{pbm}
\subsection{Data Compression}\label{Sec:NN_compression}
For cooperation, agents transmit their planned state trajectories, $x_{t,t+N_p^i}^i \in {\mathbb{R}^{{n^i} \times N_p^i}}$, but reception occurs after some delay $\Delta^{ji}$. To reduce packet size, trajectory containing $n^i\times{N^i_p}$ floating points is compressed by approximating it with neural network $\mathcal{N}^i$ of $q^i$ weights and biases, with compression factor of $1-(q^i+\mathrm{overhead}\;\mathrm{size})/({n^i \times N^i_p})$. Overhead size accounts for agent identity $i$, time-stamp ($T_s^i$) and sampling time $T^i$ etc. The leader also communicates formation geometry and way-points to followers. It is assumed that there exists a mechanism for synchronizing clocks, which allows delay $\Delta_{ji}$ to be estimated. NN at $A^i$ is trained using state trajectory as output and $N^i_p$ discrete instants as input. Using sampling rate $T^j$ and prediction horizon $N^j_p$ at $A^j$, re-sampled trajectory $\tilde{w}_t^j \in {W^j} \subset {\mathbb{R}^{{n^j} \times N_p^j}}$ is generated using received neural network $\mathcal{N}^i$. If horizon is sufficiently long, states can be extrapolated with bounded error. If packet is delayed by more than a threshold $\bar \Delta$, the packet is deemed to be lost. Any smooth function $w(t)$ can be approximated arbitrarily closely on a compact set using a NN with appropriate weights and activation functions \cite{Jagannathan2006}. Let $w(\tau)$ be a set of smooth functions, then we can show $\tilde{w}(\tau)=w(\tau)+\xi$, where $\tilde{w}(\tau)$ is  approximation of $w(\tau)$ by NN, and $\tau\triangleq{col(t,t\dots{t})}$ is the stack of $t$ vector $n^i$ times and $\xi$  is NN approximation error which is inversely proportional to hidden-layer size $H_L$. Error $\xi^i_t$ in prediction also depends on the delay $\Delta_t^{ij}$ in information received from $A^j$ due to extrapolation of trajectory tail ($\tilde w^{i,t+N_p^i+\Delta^{ij}}_{t+N_p^i}$). If the error (or delay) is greater than an upper bound, i.e. $\xi^i_t>\bar \xi$, a feasible control for avoiding collision may not exist. This means that agents will get too close due to error $\xi^i_t$, such that there is not enough time to maneuver for avoiding collision. Consequently, we assume an upper bound on the permissible delay $\Delta_t^{ij}\le\bar\Delta$, which is the worst case scenario of two agents on a direct collision course at maximum permissible speed and with minimum separation between them, i.e. $\bar \Delta  \triangleq {R_{\min }}/{v_{\max }}$. With this conservative (can be relaxed) bound on $\Delta_t^{ij}$, there is always enough time to execute collision avoidance maneuvers. 
\alglanguage{pseudocode}
\begin{algorithm}
\caption{DNMPC Algorithm with Collision Avoidance}\label{Alg:dnmpc}
\begin{algorithmic}[1]
\State {\bf Given} {$A^1$,$A^i\gets x_0^i, d^{h^i},d^{q^i}, g^i $}\Comment{$i=1\triangleq$ Leader, $t=0$}\label{Al:lead_sel_dnmpc}
\State {\bf Solve} Problem \ref{Pmb:Terminal_Set_Ctrl_Convex} offline for $Q_f^i$ and $K_f^i$\label{Al:convex_term_ctrl_set}
\Procedure{Collision Free Distributed NMPC}{}
        \State \textbf{Design} Spatially filter potential \eqref{E:colav_weight_condition}
        \State {\bf Solve} Problem \ref{P:RHOCP} at $A^i$  for  $u^{i,\star}_{t,t+N_c^i-1}$ \label{Al:optimize_dnmpc}
   \State {\bf Train NN} Train Neural network for $x^{i,\star}_{t,t+N_p}$ \label{Al:train_NN_dnmpc}
   \State {\bf Implement} first element/block of $u^{i,\star}_{t,t+N_c^i-1}$ \label{Al:implement_dnmpc}
   \State {\bf Transmit/Receive} data packets \label{Al:send}
   \State{\bf Estimate} time delay $\Delta_{ij}$
   \State {\bf Reconstruct} $\tilde{w}_t^{t+Np^i}$  with received NN
   \Statex {\bf Increment} time by one sample \Comment{$t^i=t^i+T_s^i$}
\EndProcedure\Comment{End CF-DNMPC Alg.}
\end{algorithmic}
\end{algorithm}

\section{Stability Analysis }\label{Sec:stab_analysis}
We first state an important new result in regional input-to-state practical stability. This general result will form the cornerstone of later development. 
\newtheorem{thm}{Theorem}
\begin{thm}\label{T:ISpS_gen}
If system $x_{t+1}=f(x_t,w_t)$ admits an ISpS-Lyapunov function in $\Xi$, then it is regional ISpS and satisfies condition (\ref{E:ISpS_def}), with $\beta (r,s) \triangleq {\alpha _1}^{ - 1} (3\hat\beta (3{\alpha _3}\left( r \right),s))$, 
 $\gamma (s)\triangleq {\alpha _1}^{ - 1}(3(\hat \gamma (3\sum\limits_{i = 1}^3 {{\sigma _i}(s)} ) + \hat \beta (3{\sigma _3}\left( s \right),0)))$ and $c \triangleq {\alpha _1}^{ - 1}(3(\hat \beta (3(\bar{\bar c}+d),0) + {\alpha _1}^{ - 1}\hat \gamma (\mu (3\bar {\bar c})) + {\alpha _1}^{ - 1}\hat \gamma (3\bar c))$, 
where $\mu$, $\hat{\gamma}\in\mathcal{K}_\infty$ while $\hat{\beta}\in\mathcal{KL}$ ($d$ is defined in proof).
\end{thm}
\begin{proof}
for all ${w_{t,t+1}} \in W$. Since $\Omega$ is RPI, therefore for $x_1\in\Xi\backslash\Omega$ and $x_2\in\Omega$, there  exists $d>0$ such that $ V(x_1,w_1) \le V(x_2,w_2) + d$ for $w_{1,2}\in W$.
Letting $\bar{\alpha}_3(s)\triangleq\alpha _3(s) + {\sigma _3}(s)+s$, $\underline\alpha_2(s) \triangleq \min ({\alpha_2}(s/3) , {\sigma_3}(s/3) , \mu (s/3))$, $\alpha_4(s)\triangleq\underline\alpha_2(s)\circ\bar\alpha^{-1}_3(s)$, $\hat{w}\triangleq\max(|w_t
|_\infty,|w_{t+1}|_\infty)$, $\omega(\hat{w},\bar c,\bar{\bar{c}})\triangleq \sum\limits_{i = 1}^3 {{\sigma _i}(\hat w)}+ \mu (\bar {\bar c}) + \bar c $ and selecting 
$\rho\in\mathcal{K}_\infty$ such that $(\mathcal{I}-\rho)\in\mathcal{K}_\infty$, we can define a compact set $D\subset\Omega\subset\Xi$ containing the origin: $D\triangleq\{ x|\,\,d(x,d\Omega ) > {d_1},\,\,V({x_t},{w_t}) \le \hat{\gamma}(\omega)\}$, where $\hat{\gamma}\triangleq \alpha_4^{-1}\circ\rho^{-1}$. With these definitions and using steps similar to equations (14)-(17) in proof of Theorem 4.1 of \cite{Franco2008}, we can show that $D$ is RPI. Moreover, $D$ can also be shown to be asymptotically attractive  for state starting in $\Xi \backslash D$ using arguments similar to equations (18)-(23) of \cite{Franco2008}. Hence, a state $x_t$ starting in $\Xi$ will enter $\Omega\backslash D$ in finite time, and from there it will enter $D$ in finite time as well, where it shall remain as $D$ is RPI. Using a standard comparison lemma \cite{Jiang2001ISS}, $\exists\hat\beta(r,s)\in\mathcal{KL}$ such that $V\left( {{x_t},{w_t}} \right) \le \max (\hat{\beta} (V\left( {{x_0},{w_0}} \right),t),\hat \gamma (\omega (|{w_t}{|_\infty },\bar c,\bar {\bar c})),\,\forall x_t \in \Xi, w_t\in W$. Using a property for $\mathcal{K}$ functions: $\alpha ({r_1} + {r_2} + {r_3}) \le  \alpha (3\max ({r_1},{r_2},{r_3}))\le \alpha (3{r_1})+\alpha (3{r_2})+\alpha(3{r_3})$, 
we can show that system $x_{t+1}=f(x_t,w_t)$ is regional ISpS in $\Xi,  \forall x_t\in\Xi, w_t\in W$.
\end{proof} 
We will now particularize this result for Algorithm \ref{Alg:dnmpc}. Stability is analyzed in two stages. First, individual agents are shown to be ISpS and robust to communication delays and trajectory approximation error in a subset of $X^i$, followed by generalized small gain condition for team stability.
\subsection{Stability of Individual Agents without Collision Avoidance} \label{Sec:Stab}
Asymptotic stability (ISS) for MPC schemes can be shown in case of additive and vanishing disturbance, but only ultimate boundedness (or ISpS) can be guaranteed in case of non-vanishing (not decaying with state) uncertainties \cite{Limon2006minmaxMPC}. In the proposed approach, the uncertainty in trajectory approximation $\xi$ is non-vanishing and  one can only guarantee ISpS. We consider first the stability of individual agent $A^i$ with respect to the information received from other agents, by exploiting Theorem \ref{T:ISpS_gen}. At this stage the interconnections are ignored, and information from neighbors is considered as external input. We assume at this stage that agents generate conflict free trajectories. 
\newtheorem{assump}{Assumption}
\begin{thm}\label{T:ISpS_spec}
Let terminal set $X_f^i\subset X^i$ be RPI and let $k_f^i(x_t^i), f^i(x_t^i,k^i_f(x_t^i)), w^i_{t+1}, h^i(x_t^i,u_t^i), q^i(x_t^i,w_t^i)$ $h_f^i(x_t^i,u_t^i)$  be locally Lipschitz with respect to $x_t^i,u_t^i$ and $w_t^i$ in $X^i\times U^i\times W^i$,  with the following Lipschitz constants $L^i_{k_f}, L^i_f, L_{gw}, L_{hx}^i, L_{hu}^i, L_{qx}^i, L_{qw}^i$ and $L_{hf}^i$. Moreover, there exist nonlinear bounds $\alpha_{1,f}, \alpha_{2,f}, \underline{r}^i\in\cal{K}_\infty$ such that 
$\underline{r}^i(|x_t^i|)\le h^i(x_t^i,u_t^i)$ and $\alpha_{1,f}(|x_t^i|)\le{h_f^i(x_t^i)}\le{\alpha_{2,f}(|x_t^i|)}, \forall x_t^i\in{X^i}$. Now, if the neural network trajectory approximation error is bounded $|\tilde{w}_t|\le{|w_t|+\hat{\xi}}$, and the following holds for $x_t^i\in X_f^i$ and $w_t^i\in W^i$
\begin{equation}\label{E:terminal_constraint_condition}
	h_f^i(f^i(x^i,k_f^i(x^i)))-h_f^i(x^i)
	\le-h^i(x^i,k_f^i(x^i))-q^i(x^i,\tilde{w}^i)+{\psi^i(|\tilde{w}^i|)}
\end{equation}
for some $\psi^i\in\cal{K}$, then agent $A^i$ under NMPC optimal $u^{i,\star}$ and terminal $k_f^i(x^i)$ control laws admits ISpS Lyapunov function $V(x^i_t,{w}^i_t,u^i_t)=J^i_t(x^{i,\star}_t,{w}^i_t,u^{i,\star}_{t,t+N_p^i})$ and is therefore ISpS with robust output feasible set $X^i_{MPC}\subseteq X^i$, which is the set of initial states for which the Problem \ref{P:RHOCP} is feasible. 
\end{thm}
\begin{proof}
We need to prove that $V(x^i_t,u^i_t,{w}^i_t)=J^i_t(x^{i,\star}_t,u^{i,\star}_{t,t+N_p^i},{w}^i_t)$ is an ISpS Lyapunov function. The lower bound on $V(x^i_t,{w}^i_t)$ is obviously given by $
\underline{r}^i(|x_t^i|)=\alpha^i_1 (|x_t^i|)\le V(x^i_t,{w}^i_t),\,\,\,\ \forall x^i_t\in X^i, w^i_t\in W^i$. 
Local control $\tilde{u}^i_{t,t+N_c^i-1}={[k_f^i(x_t^i),\dots,k_f^i(x^i_{t+N_c^i-1})]}^T$ is feasible but suboptimal $\forall x^i\in  X_f^i$, i.e. $V(x_t^i,\tilde w_t^i) \le J_t^i(x_t^i,\tilde w_t^i,\tilde u_{t,t + N_p^i}^i)$. Using assumptions in Theorem \ref{T:ISpS_spec}, we get
$V(x_t^i,w_t^i) \le\alpha _3^i(| {x_t^i}|) + \sigma _3^i(| {w_t^i}|) + \bar {\bar{c}}^i$, where  $\alpha_3^i(s)={\alpha^i_{2,f}}(L{{_f^i}^{N_p^i}} s)+\bar{b}^i$, $\sigma_3^i(s)=\bar{\bar{b}}^i s$ and $\bar{\bar{c}}^i={N_p^i(L^i_{qw})}|\hat{\xi}^i|$. The constants are $\bar{b}^i={(L_h^i + L_{hu}^iL_{{k_f}}^i + L_q^i)(L{{_f^i}^{N_p^i}} - 1)}(L_f^i - 1)^{-1}$ and $\bar{\bar{b}}^i=L_{qw}^i({L_{gw}^i}^{N_p^i}-1)(L_{gw}^i - 1)^{-1}$. Clearly, $\tilde{\tilde{u}}^i_{t+1,t+N_c^i}={[u^{i,\star}_{t+1,t+N_c^i-1},k_f^i(x^i_{t+N_c^i})]}^T$ is also feasible control for $x^i\in  X_{MPC}^i$ which gives $V(x^i_{t+1},w^i_{t+1})\le\sum\limits_{l = t + 1}^{t + N_p^i} {\{ h(x_l^i,\tilde{\tilde{ u}}_l^i) + q({x_l},\tilde {\tilde{w}}_l^i)\} } + h_f^i({f^i}(x_{t + N_p^i}^i,k_f^i(x_{t + N_p^i}^i))$, where, $\tilde {\tilde{w}}_l$ is NN approximation of $w^i_{t+1}$ and $\tilde{w}_l$ is approximation of $w^i_{t}$, hence $\tilde {\tilde{w}}_l\ne\tilde{w}_l$. Canceling common terms, we get 
$V(x_{t + 1}^i,w_{t + 1}^i) - V(x_t^i,w_t^i) \le-\alpha _2^i(|x_t^i|)+\sigma _1^i(|w_t^i|)+\sigma _2^i(|w_{t + 1}^i|)+\bar{c}^i$, where $\sigma_1(s)=\sigma_2(s)+\psi^i(s)$, $\sigma_2(s)=\underline{b} s$, $\bar{c}=\sigma_2(\hat{\xi})$ and $\underline{b}=L_{qw}^i({L_{gw}^i}^{N_p^i-1}-1)(L_{gw}^i - 1)^{-1}$. Hence, from Theorem \ref{T:ISpS_gen}, the system \eqref{E:agentsys} under NMPC is ISpS.
\end{proof}
  A method for terminal control law design (by solving \eqref{E:terminal_constraint_condition}) is given: Let $h^i_l={x^i_l}^T{Q^i}{x^i_l}+{u^i_l}^T{R^i}{u^i_l}$, $q^i_l\le{x^i_l}^T{S^i}{x^i_l}+\psi{(|\tilde{w}^i|)}$ and $h^i_f={x^i_f}^T{Q_f^i}{x^i_l}$ , where $Q^i$, $R^i$, $Q_f^i$ and ${S}^i$ are positive definite matrices for $i=1,\dots M$ agents. Let $k_f^i(x^i_l)=K^ix^i_l$ exist, such that $A_{c}^i=A_o^i+B_o^iK^i$ is stable, where $A_o^i$ and $B_o^i$ are the Jacobians of system \eqref{E:agentsys}. The terminal set is defined as $X_f^i\triangleq (x^i)^T Q_f^i x^i\le a$ for some $a\in \mathbb{R}_{\ge 0}$ which satisfies constraints $x^i \in {X^i} $ and $u^i=K^i x^i \in {U^i}$. Let $Q_f^i$ be the solution of the convex problem.
\begin{pbm}\label{Pmb:Terminal_Set_Ctrl_Convex}
	\begin{equation}
	\mathop {\min }\limits_{{Q_f},{K_f}} \,\left[ { - \log \left( {\det \left( a Q_f^i\right)} \right)} \right]
	\end{equation}
	subject to the Lyapunov inequality
	$A_c^{{i^T}}{Q_f^i}A_c^i - {Q_f^i} + {Q^i} + K^{{i^T}}{R^i}K^i + {(N-1) S^i}\preceq 0$, and $Q_f>0$
\end{pbm}
\subsection{Stability of Individual Agents with Collision Avoidance}\label{Sec:Collav}
Results of the previous section will now be extended to prove stability of the agents under the collision avoidance scheme described in Section \ref{Sec:DNMPC-CA}. Let  $V(x^i_t,{w}^i_t)=J^i_t(x^{i,\star}_t,{w}^i_t)$ be the local ISpS Lyapunov function for $A^i$ without collision avoidance. Let $x^{i,\star}_{t,t+N_p^i}$ be the optimal solution of the cost \eqref{E:cost1} and $\acute x^{i,\star}_{t,t+N_p^i}$ be the optimal solution of the modified cost \eqref{E:colav_cost}. We will prove that $\acute V(\acute x^i_t,{w}^i_t)=J^i_t(\acute x^{i,\star}_t,{w}^i_t)$ is an ISpS Lyapunov function. It is obvious that $d_{ij}(k)\ne 0$ for at least one instant $t\le k \le t+N_p^i$, since otherwise would mean that the current position as well planned optimal trajectories of two agents coincide exactly, which is impossible. We assume that $\underline \kappa^i |\acute x^{i,\star}| \le |x^{i,\star}|\le \overline \kappa^i |\acute x^{i,\star}|$, for some constants $\underline \kappa^i, \overline \kappa^i\ge 0$, since both $ x^{i,\star}$ and $\acute x^{i,\star}$ are finite. This leads to bounds on potential function, i.e.  ${\underline \Phi  ^i} \le \Phi _t^i \le {\overline \Phi  ^i}$ for some constants ${\underline \Phi  ^i}, {\overline \Phi  ^i}\ge 0$.
\begin{thm}\label{T:Collav}
For an agent on collision course, the optimal trajectory $\acute x^{i,\star}_{t,t+N_p^i}$ for modified cost \eqref{E:colav_cost} not only guarantees collision avoidance with other agents in the sense of \eqref{E:colav_exp_cond}, but also maintains  input-to-state practical stability, if its repulsive spatial filter weights $\lambda(d^{ij}_{k|t})$ are chosen at each instant such that 
\begin{equation}\label{E:colav_weight_condition}
\frac{{{\lambda^i_{\max ,t}}}}{{{\lambda^i_{\min ,t}}}} < \frac{{{\underbar{r}^i}(|{x_t}|){{\{ N_p^iR_{\min }^i + N_p^i(N_p^i - 1){v_{\max }}\} }^{ - 1}}}}{{(N_p^i - 1)(L_{hx}^i + L_{qx}^i) + {L_{hf}}}}\triangleq\bar{a}_t
\end{equation}
\end{thm}
\begin{proof}
The proof consists of two parts. We first show that negative gradient of modified cost \eqref{E:colav_cost} lies in the direction of expanding weighted average distance $\bar{d}^{ij}_t$ between agents on collision course. Hence, the optimal trajectory  $\acute x^{i,\star}_{t,t+N_p^i}$ reaches the terminal set by avoiding collision in the sense of \eqref{E:colav_exp_cond}. Next, we will show that the optimal trajectory in that direction is also ISpS stable. From \eqref{E:colav_cost}, we can see that $\frac{{\partial\acute J_t^i}}{{\partial \bar d_t^{ij}}} = \frac{{\partial J_t^i}}{{\partial \bar d_t^{ij}}}(1 + \Phi _t^i) + J_t^i\frac{{\partial \Phi _t^i}}{{\partial {{\bar d}^{ij}_t}}}$. Since  $\partial \Phi_t^i/\partial\bar{d}_t^{ij}=-\Phi_t^i/\bar{d}^{ij}_t<0$ and $J_t^i, \Phi_t^i>0$, in order to have $\partial\acute J_t^i/\partial \bar d_t^{ij}<0$, we have
$\frac{{\partial J_t^i}}{{\partial \bar d_t^{ij}}} < \frac{{\Phi _t^i}}{{1 + \Phi _t^i}}\frac{{J_t^i}}{{\bar d_t^{ij}}} < \frac{{J_t^i}}{{\bar d_t^{ij}}}
$.
Since $J_t^i,d_t^{ij}>0$, this condition can be satisfied if
$\max \left| {\frac{{\partial J_t^i}}{{\partial \bar d_t^{ij}}}} \right| < \frac{{\min (J_t^i)}}{{\max (\bar d_t^{ij})}}$. For RHS, note that by chain rule of differentiation and using triangle inequality, $\left| {\frac{{\partial J_t^i}}{{\partial \bar d_t^{ij}}}} \right| \le \sum\limits_{k = t}^{t + N_p^i} {\left| {\frac{{\partial J_t^i}}{{\partial x_k^i}}} \right|} \left| {\frac{{\partial x_k^i}}{{\partial d_k^{ij}}}} \right|\left| {\frac{{\partial d_k^{ij}}}{{\partial \bar d_t^{ij}}}} \right|$. With slight abuse of notation we can write $d^{ij}_k=|x^i_k-w^i_k|$. For given neighbor trajectory $w^i_k=x^j_k, \forall j\in G^i$, we have $\partial d^{ij}_k/\partial x^i_k=(x^i_k-w^i_k)/d^{ij}_k$ such that $|d^{ij}_k/\partial x^i_k|=1$. Similarly, $\partial\bar{d}^{ij}_t/\partial d^{ij}_k=\lambda^i_k$, which results in $\left| {\frac{{\partial J_t^i}}{{\partial \bar d_t^{ij}}}} \right| \le \sum\limits_{k = t}^{t + N_p^i} {\frac{1}{{\lambda _k^i}}\left| {\frac{{\partial J_t^i}}{{\partial x_k^i}}} \right| < } \frac{1}{{\lambda _{\min ,t}^i}}\sum\limits_{k = t}^{t + N_p^i} {\left| {\frac{{\partial J_t^i}}{{\partial x_k^i}}} \right|}$. Now, from \eqref{E:cost1}, we get
\begin{eqnarray}\label{E:max_cost_grad}
\begin{array}{l}
\max \left| {\frac{{\partial J_t^i}}{{\partial \bar d_t^{ij}}}} \right| < \frac{{(N_p^i - 1)(L_h^i + L_q^i) + L_{hf}^i}}{{\lambda _{\min ,t}^i}}
\end{array}
\end{eqnarray}
Now, maximum $\bar{d}^{ij}_t$ can occur when the minimum distance between agents on collision course is $R^i_{\min}$ and then move away from each other at $v_{\max}$, i.e. $\max (\bar d_t^{ij}) = \sum\limits_{k = t}^{t + N_p^i} \lambda_k^i{(R_{\min }^i + 2(k - t){v_{\max }})} <\lambda^i_{\max,t}(N_p^iR_{\min }^i + N_p^i(N_p^i - 1){v_{\max }})$. Also, as noted in Theorem \ref{T:ISpS_spec}, $\min({J_t^i})\le V^i_t\le\underbar{r}^i(|x^i_t|)$. This can be combined with \eqref{E:max_cost_grad} to result in the condition specified in \eqref{E:colav_weight_condition}. Hence, the minimum of modified cost lies in the direction of collision avoidance in the sense of \eqref{E:colav_exp_cond}. Since any feasible trajectory for cost \eqref{E:cost1} is also feasible for modified cost \eqref{E:colav_cost} and the reachable set is compact, an optimum almost always exists, unless there is not enough time to maneuver (to cater for which we have placed a conservative bound on $\Delta^{ij}_t\le\bar{\Delta})$.

For the next part of this proof, note that $\acute J(\acute x_t^{i,\star},{w}^i_t)\le \acute J( x_t^{i,\star},{w}^i_t)$ and $ J(x^{i,\star}_t,{w}^i_t)\le  J( \acute x^{i,\star}_t,{w}^i_t)$, since $\acute x^{i,\star}_{t,t+N_p^i}$ is feasible but suboptimal control for minimization of \eqref{E:cost1} and $ x^{i,\star}_{t,t+N_p^i}$ is suboptimal for \eqref{E:colav_cost}. For conciseness, we will ignore the difference between $V$ and $J$ in this section and also drop the $\star$ symbol. From Theorem \ref{T:ISpS_spec}, we have $\alpha^i_1 (|x_t^i|)\le V(x^i_t,{w}^i_t)$, which gives $\alpha^i_1 (\underline \kappa^i |\acute x_t^i|)\le V(x^i_t,{w}^i_t)\le V(\acute x^i_t,{w}^i_t)$. Combining this with \eqref{E:colav_cost} and defining $\acute \alpha^i_1(s)\triangleq(1+\underline \Phi^i)\alpha^i_1 (\underline \kappa^i s) \in\mathcal{K}_\infty$, we get $\acute \alpha^i_1 (|\acute x_t^i|)\le \acute V(\acute x^i_t,{w}^i_t)$. Let $\ V(\acute x^i_t,{w}^i_t)- V(x^i_t,{w}^i_t) \le \varkappa^i$ for some constant $\varkappa^i>0$. Defining $\acute \alpha^i_3(s)\triangleq (1+\bar \Phi^i)\alpha^i_3 (\bar \kappa^i s)\in\mathcal{K}_\infty$, ${\acute \sigma _3}(s)\triangleq(1+\bar \Phi^i)\sigma _3(s)\in\mathcal{K}$ and $ \acute{\acute{c}}^i\triangleq(1+\bar \Phi^i)(\bar{\bar{c}}^i+\varkappa^i)$, we get $\acute V\left( {\acute x_t^i,w_t^i} \right)  \le {\acute \alpha _3}( | {\acute x_t^i} | ) + {\acute \sigma _3}\left( {\left| {w_t^i} \right|} \right) + \acute{ \acute{ c}}^i$. Using \eqref{E:colav_cost}, and defining $\acute \alpha^i_2(s)\triangleq(1+\underline \Phi^i)\alpha^i_2 (\underline \kappa^i s) \in\mathcal{K}_\infty$, ${\acute \sigma _{1,2}}(s)\triangleq(1+\bar \Phi^i)\sigma _{1,2}(s)\in\mathcal{K}$, $ {\acute{c}}^i\triangleq(1+\bar \Phi^i)({\bar{c}}^i+ \varkappa^i)$, we get $\Upsilon^i_{t+1} \acute V\left( {\acute x_{t + 1}^i,w_{t + 1}^i} \right) - \acute V\left( {\acute x_t^i,w_t^i} \right) \le - {\acute \alpha _2}\left( { \left| {\acute x_t^i} \right|} \right) + {\acute \sigma _1}\left( {\left| {w_t^i} \right|} \right) + {\acute \sigma _2}\left( {\left| {w_{t + 1}^i} \right|} \right) + {{\acute c}^i}$, where,$\Upsilon _{t + 1}^i \triangleq \frac{{1 + \Phi _{t+1}^i}}{{1 + \Phi _{t}^i}}$. From \eqref{E:colav_pot}, $\Upsilon^i_{t+1}\ge 1$ if  \eqref{E:colav_exp_cond} holds and we can write $
 \acute V\left( {\acute x_{t + 1}^i,w_{t + 1}^i} \right) - \acute V\left( {\acute x_t^i,w_t^i} \right) \le - {\acute \alpha _2}\left( { \left| {\acute x_t^i} \right|} \right) + {\acute \sigma _1}\left( {\left| {w_t^i} \right|} \right) + {\acute \sigma _2}\left( {\left| {w_{t + 1}^i} \right|} \right) + {{\acute c}^i}$. Hence, agent $A^i$ is ISpS according to Theorem \ref{T:ISpS_gen} and moves towards its goal in an optimal manner while avoiding collision with other agents.
\end{proof}
\newtheorem{corol}{Corollary}
\begin{corol}\label{Corol:GP Spatial Filter}
If spatial filter for collision avoidance is shaped as a geometric progression  $\lambda^i_{k|t}=\lambda^i_{\max,t}r_t^l$ such that $d^{ij}_l>d^{ij}_{l+1}$ for   $l=0,\dots N_p^i-1$, then the filter can be designed by specifying $\bar{b}>1$, $\lambda^i_{\max,t}$ and calculating  $r_t={(\bar{b}\bar{a}_t)}^{{-1}/{(N_p^i-1)}}$ from \eqref{E:colav_weight_condition}. 
\end{corol}
\subsection{Stability of Team of Agents under NMPC}\label{Sec:TeamNet}
We will establish a generalized small gain condition to prove stability of the interconnected system, for both strongly- and weakly-connected network topologies. The result is general, not limited by the number of subsystems and the way in which subsystem gains are distributed is arbitrary. 
\begin{thm}\label{T:SGC}
For a team of agents $A^i$ \eqref{E:sys_cl}, each with local ISpS Lyapunov function $V(x^i_t,w^i_t)$, there exists  $\bar{\alpha}\in\mathcal{K}_\infty$ such that $V(x_{t + 1}^i,w_{t + 1}^i) - V(x_t^i,w_t^i) \le \bar{\alpha}^i(|x^i_t|)$. Let the ISpS Lyapunov gain from $A^i$ to $A^j\in G^i$ be denoted by the function $\bar\gamma_{ij}(s):\mathcal{R}_{\ge 0}\to\mathcal{R}_{\ge 0}$ and given by
\begin{equation}\label{E:nlgains}
\bar\gamma_{ij}(s)\triangleq \alpha_1^i \circ (\bar{\alpha}^i)^{-1} \circ \sigma_1^i \circ ({\alpha_1^j})^{-1}(s),
\end{equation} 
then the team of agents is ISpS stable if the network is at least weakly connected, as long as the following small gain condition is satisfied
\begin{equation}\label{E:SGC_expanded}
V(x_t^i,w_t^i) > \mathop {\max}\limits_{j \in {G^i},j \ne i} \{{\bar\gamma _{ij}}(V(x_t^j,w_t^j))\}
\end{equation}
\end{thm}
\begin{proof}
Consider $\bar{\rho}^i\in\mathcal{K}_\infty$.  Let $V(x_{t + 1}^i,w_{t + 1}^i)-V(x_t^i,w_t^i)\le-\alpha_2^i(|x_t^i|)+\sigma_1^i(|w_t^i|)+\sigma_2^i(|w_{t+1}^i|)+\bar{c}^i\le\bar{\rho}^i\circ\alpha_2^i(|x_t^i|)$ for $x^i_t\in X^i\backslash\mathcal{B}^n(c^i)$ and for  $\bar{\rho}^i\in \mathcal{K}_\infty$ constructed such that $\sigma _1^i(|w_t^i|) + \sigma _2^i(|w_{t + 1}^i|) + {\bar c^i} \le ({\cal I} + {\bar \rho ^i})^\circ \alpha _2^i\left( {|x_t^i|} \right)$. Then, in view of \eqref{E:lyapISS_cond1} and letting $\bar{\alpha}^i\triangleq(\mathcal{I}+\bar{\rho} ^i)\circ \alpha _2^i$, we get $V(x_t^i,w_t^i) \ge \alpha _1^i \circ {(\bar{\alpha}^i)^{ - 1}} \circ \sigma _1^i(|w_t^i|)$. Now, since $w^i_t=col(x^j_{t,t+N^j_p})$, then $|w_t^i| \ge \mathop {\max }\limits_j |x_t^j| \ge |x_t^j|$, $\forall j\in G^i$. Hence, $V(x_t^i,w_t^i) \ge \mathop {\max }\limits_j (\alpha _1^i \circ {(\bar\alpha^i)^{ - 1}} \circ \sigma _1^i(|x_t^j|))$. But, $V(x_t^j,w_t^j) \ge \alpha _1^j (|x_t^j|)$ $\Rightarrow  {(\bar\alpha^j)^{ - 1}}(V(x_t^j,w_t^j))\ge |x_t^j|$, and hence $V(x_t^i,w_t^i) \ge \mathop {\max}\limits_j (\alpha _1^i \circ {(\bar\alpha^i)^{ - 1}} \circ \sigma _1^i \circ {(\alpha _1^j)^{ - 1}}(V(x_t^j,w_t^j))$. If gain $\bar\gamma_{ij}$ is defined as in \eqref{E:nlgains}, then \eqref{E:SGC_expanded} is obtained. From recent results in  \cite{Dashkovskiy2010}, it can be shown that this is equivalent to having an ISpS Lyapunov function for the network. 
\end{proof}
\subsection*{Remark 1}
	One way to design $\bar{\alpha}^i$ is by choosing  $\bar{\rho}^i(s)=\bar{k}^i s, \forall \bar{k}^i>0$, since it was shown that $V^i_{t+1}-V^i_{t}<0$. This choice results in stable network, provided that individual agents are locally ISpS. We take the case of agents not on collision course first.
	\subsubsection*{Agents not on collision course}
	Continuing from proof of Theorem 2, and letting ${\lambda _\Pi }_{\max }$ and ${\lambda _\Pi }_{\min }$ be the maximum and minimum eigenvalues of a p.d. matrix $\Pi$, respectively. Then,
	\begin{equation*}
	\sigma _1^i\left( r \right) = \sigma _2^i\left( r \right) + {\psi ^i}\left( r \right) = \frac{{L_{qw}^i\left( {L{{_{gw}^i}^{N_p^i - 1}} - 1} \right)}}{{L_{gw}^i - 1}}r + \left( {M - 1} \right){\lambda _{S_{max}^i}}{r^2},\;\forall L_{gw}^i \ne 1,
	\end{equation*}
	For $L_{gw}=1$, the results need trivial modifications, by replacing $L_{qw}^i\left( {L{{_{gw}^i}^{N_p^i - 1}} - 1} \right){\left( {L_{gw}^i - 1} \right)}^{ - 1}$ with $\mathop \sum \limits_{l = 0}^{l = N_p^i - 2} {L{_{gw}^i}^l} = N_p^i - 1$. Similarly,
	\begin{equation*}
	\mathop {\alpha _1^j}\nolimits^{ - 1} (r) = \sqrt {\frac{r}{{{\lambda _{\min {Q^j}}}}}} ,
	\end{equation*}
	\begin{equation*}
	{\bar \alpha ^{{i^{ - 1}}}}\left( r \right) = \alpha _2^{{i^{ - 1}}} \circ {\left( {I + {{\bar \rho }^i}} \right)^{ - 1}}\left( r \right)
	\end{equation*}
	and
	\begin{equation*}
	{L_{qw}} = {\lambda _S}_{max}|\tilde w_{max}^i|
	\end{equation*}
	We mentioned that one choice of ${\bar \rho ^i}$ could be ${\bar \rho ^i}\left( r \right) = {\bar k^i}r$ for all ${\bar k^i} > 0$. Therefore,
	\begin{equation*}
	{\bar \alpha ^{{i^{ - 1}}}}(r) = {\bar \alpha _2}^{{i^{ - 1}}}\left( {\frac{1}{{{{\bar k}^i} - 1}}r} \right)
	\end{equation*}
	It is also worth noting that we showed in the proof of Theorem \ref{T:ISpS_spec} that $\alpha_2(r)=\alpha_2(r)$. Since ${{\bar \gamma }_{ij}}(r) = \alpha _1^i \circ {{\bar \alpha }^{{i^{ - 1}}}} \circ \sigma _1^i \circ \alpha _1^{{j^{ - 1}}}(r)$, we can obtain
	\begin{equation*}
	{\bar \gamma _{ij}}\left( r \right) = \frac{1}{{{{\bar k}^i} + 1}}\left( {N_p^i - 1} \right){\lambda _{S_{max}^i}}\tilde w_{max}^i\sqrt {\frac{r}{{{\lambda _{min{Q^j}}}}}}  + \frac{1}{{{{\bar k}^i} + 1}}\left( {M - 1} \right){\lambda _{S_{max}^i}}\left( {\frac{r}{{{\lambda _{min{Q^j}}}}}} \right)
	\end{equation*}
	Hence, \eqref{E:SGC_expanded} can be written as
	\begin{multline*}
	{{\bar \gamma }_{ij}}\left( {V\left( {{x^i},{{\tilde x}^j}} \right)} \right) = \left[ {\frac{1}{{{{\bar k}^i} + 1}}\left( {N_p^i - 1} \right)\;\;\tilde w_{max}^i\frac{{\lambda _{{S_{max}}}^i}}{{\sqrt {{\lambda _{min{Q^j}}}} }}} \right]\sqrt {V\left( {{x^i},{{\tilde x}^j}\;} \right)} \\
	+ \left[ {\frac{1}{{{{\bar k}^i} + 1}}\left( {M - 1} \right)\left( {\frac{{{\lambda _{S_{max}^i}}}}{{{\lambda _{min{Q^j}}}}}} \right)} \right]V\left( {{x^i},\;{{\tilde x}^j}} \right)
	\end{multline*}
	Therefore, by choosing a suitable value of ${{{\lambda _{S_{max}^i}}} \mathord{\left/
			{\vphantom {{{\lambda _{S_{max}^i}}} {{\lambda _{min{Q^j}}}}}} \right.
			\kern-\nulldelimiterspace} {{\lambda _{min{Q^j}}}}}$ and $\bar{k}^i>0$, the small gain condition \eqref{E:SGC_expanded} can be satisfied. 
\subsubsection*{Agents on collision course}
For agents on collision course, similar results can be reproduced as all functions have corresponding counterparts in collision avoidance case, see proof of Theorem \ref{T:Collav}. Therefore we can write \eqref{E:nlgains} as $\bar\gamma'_{ij}(r)\triangleq \alpha_1^{'i} \circ (\bar{\alpha}^{'i})^{-1} \circ \sigma_1^{'i} \circ ({\alpha_1^{'j}})^{-1}(r)$. Thus, we get
\begin{equation*}
{\bar \gamma' _{ij}}\left( r \right) = \frac{{\left( {1 + {{\bar \phi }^i}} \right)}}{{\left( {{{\bar k}^i} + 1} \right){\underline{\kappa} ^j}}}\left( {\left( {N_p^i - 1} \right)\;\;|\tilde w_{max}^i|\frac{{\lambda _{{S_{max}}}^i}}{{\sqrt {\left( {1 + {\underline\phi ^j}} \right){\lambda _{min{Q^j}}}} }}{{\left( r \right)}^{\frac{1}{2}}} + \frac{{\left( {M - 1} \right){\lambda _{S_{max}^i}}}}{{{\underline\kappa ^j}\left( {1 + {\underline\phi ^j}} \right){\lambda _{min{Q^j}}}}}r} \right)
\end{equation*}
Hence, even with collision avoidance, it possible to find $\bar{k}^i>0$ which satisfies the small gain condition (see \ref{Sec:Sim_res}).	
As far as the small gain condition for weakly connected networks is concerned, we show in Remark 2 that the small gain condition is equivalent to that for strongly connected networks. It should be noted that there is no need to find the exact numerical values for construction of controller. As long as there exists some $\bar{k}^i>0$, we can be assured of ISpS stability of the team. See Section \ref{Sec:Sim_res}.

\subsubsection{Strongly Connected Network}
We will now particularize the result of Theorem \ref{T:SGC} to the case of strongly connected network.
\newtheorem{lem}{Lemma}
\begin{lem}
	A team of $N$ agents connected with a strongly connected network is ISpS stable if each agent $A^i$ has an ISpS Lyapunov function $V(x^i_t,w^i_t)$, edge gain $\bar\gamma_{ij}$ is defined as in \eqref{E:nlgains} and the following small gain condition is achieved:
	\begin{equation}
	V(x_t^i,w_t^i) > \mathop {\max }\limits_{j} ({\bar\gamma _{ij}}(V(x_t^j,w_t^j))), \,\,\, \forall j \ne i,j=1,\dots,N-1
	\end{equation}
\end{lem}
\begin{proof}
	If $\bar\mu^i$ is a monotone aggregation function (MAF) and $\Gamma$ is its irreducible gain matrix, define the gain operator $\Gamma_{\bar\mu}: \mathbb{R}^n_+ \to \mathbb{R}_+^n$,
	\begin{equation}
	{\Gamma _{\bar\mu} }\triangleq\left[ {\begin{array}{*{20}{c}}
		{{r_1}}\\
		\vdots \\
		{{r_n}}
		\end{array}} \right] \mapsto \left[ {\begin{array}{*{20}{c}}
		{{\mu _1}({\bar\gamma _{12}}({r_1}), \ldots ,{\bar\gamma _{1n}}({r_n}))}\\
		\vdots \\
		{{\mu _n}({\bar\gamma _{n,1}}({r_1}), \ldots ,{\bar\gamma _{n,n - 1}}({r_{n - 1}}))}
		\end{array}} \right]
	\end{equation}
	According to the recent generalized small gain theorems of \cite{Dashkovskiy2010}, if a strongly connected network obeys the following small gain condition (SGC): $\mathcal{I}>\Gamma_{\bar\mu}$, then it is stable in the ISS sense (see Theorem 5.3 of \cite{Dashkovskiy2010}). Now, $\bar\mu=\max$ is a monotone aggregation function (\cite{Jiang2011Book}). Let $r=(V(x^i_t,w^i_t),\dots,V(x^N_t,w^N_t))$, then the SGC is satisfied if:
	\begin{equation*}
	\begin{array}{l}
	\left[ {\begin{array}{*{20}{c}}
		{V(x_1^t,w_1^t)}\\
		\vdots \\
		{V(x_N^t,w_N^t)}
		\end{array}} \right] >
	\left[ {\begin{array}{*{20}{c}}
		{\max ({\bar\gamma _{12}}(V(x_2^t,w_2^t)), \ldots ,{\bar\gamma _{1,N}}(V(x_N^t,w_1^t)))}\\
		\vdots \\
		{\max ({\bar\gamma _{N,1}}(V(x_1^t,w_1^t)), \ldots ,{\bar\gamma _{N,N - 1}}(V(x_{N - 1}^t,w_{N - 1}^t)))}
		\end{array}} \right]
	\end{array}
	\end{equation*}
	This can be simply stated as:
	\begin{equation*}
	V(x_t^i,w_t^i) > \mathop {\max }\limits_{j} ({\bar\gamma _{ij}}(V(x_t^j,w_t^j))), \,\,\, \forall j \ne i,j=1,\dots,N-1
	\end{equation*}
\end{proof}
\subsubsection{Weakly Connected Network}
We will now focus on the case of a network of agents, in which not all agents are connected to every other agent.
\begin{lem}
	A team of cooperating agents connected with a weakly connected network is ISpS stable if each agent $A^i$ has an ISpS Lyapunov function $V(x^i_t,w^i_t)$, edge gain $\bar\gamma_{ij}$ is defined as in \eqref{E:nlgains} and the following small gain condition is achieved:
	\begin{equation*}
	V(x_t^i,w_t^i) > \mathop {\max }\limits_{j} ({\bar\gamma _{ij}}(V(x_t^j,w_t^j))), \,\,\, \forall j \ne i,j\in G^i
	\end{equation*}
\end{lem}
\begin{proof}
	The connectivity gain matrix for a weakly connected network can be brought in upper block triangular form by appropriate re-indexing of agents, such that each upper block on the diagonal is either $0$ or irreducible. Hence, we can now rewrite the gain matrix as:
	\begin{equation*}
	\Gamma  = \left[ {\begin{array}{*{20}{c}}
		0&{{\bar\gamma _{12}}}&{{\bar\gamma _{13}}}& \ldots &{{\bar\gamma _{1,\bar M}}}\\
		0& \ddots &{{\bar\gamma _{23}}}& \ldots &{{\bar\gamma _{2,\bar M}}}\\
		\vdots & \ddots & \ddots & \ddots & \vdots \\
		{}&{}& \ddots & \ddots &{{\bar\gamma _{N - 1,\bar M}}}\\
		0& \ldots &{}&0&0
		\end{array}} \right]
	\end{equation*}
	where $\bar{M}\triangleq\max\limits_{i}M^i$ is the size of neighborhood of the most connected agent. According to Proposition 6.2 of \cite{Dashkovskiy2010}, the interconnected system is stable if each upper diagonal block satisfies the SGC: $\mathcal{I}>\Gamma_{\bar\mu}$,. Now, the upper diagonal blocks are:
	\begin{equation*}
	\begin{array}{l}
	{\overline \Gamma  _1} = 0,\,\,\,{\overline \Gamma  _2} = \left[ {\begin{array}{*{20}{c}}
		0&{{\bar\gamma _{12}}}\\
		0&0
		\end{array}} \right],\,\,{\overline \Gamma  _3} = \left[ {\begin{array}{*{20}{c}}
		0&{{\bar\gamma _{12}}}&{{\bar\gamma _{13}}}\\
		0&0&{{\bar\gamma _{23}}}\\
		0&0&0
		\end{array}} \right]\\
	\\
	{\overline \Gamma  _d} = \left[ {\begin{array}{*{20}{c}}
		0&{{\bar\gamma _{12}}}&{{\bar\gamma _{13}}}& \ldots &{{\bar\gamma _{1,d}}}\\
		0& \ddots &{{\bar\gamma _{21}}}& \ldots &{{\bar\gamma _{2,d}}}\\
		\vdots & \ddots & \ddots & \ddots & \vdots \\
		{}&{}& \ddots & \ddots &{{\bar\gamma _{d - 1,d}}}\\
		0& \ldots &{}&0&0
		\end{array}} \right],\,\,\,{\overline \Gamma  _N} = \Gamma 
	\end{array}
	\end{equation*}
	Then stability is assured if each of the above blocks obey the SGC iteratively, i.e.
	\begin{equation}
	\begin{array}{l}
	{r_1} > {\overline \Gamma  _{\bar\mu 1}}\,({r_1}) \Rightarrow \,\,V(x_t^1,w_t^1) > 0\\
	{r_2} > {\overline \Gamma  _{\bar\mu 2}}({r_2}) \Rightarrow \,\,\,V(x_t^1,w_t^1) > {\bar\gamma _{12}}(V(x_t^2,w_t^2)) ,V(x_t^2,w_t^3) > 0\\
	{r_3} > {\overline \Gamma  _{\bar\mu 3}}({r_3}) \Rightarrow\,\,\,
	V(x_t^1,(x_t^1,w_t^1)) > \max ({\bar\gamma _{12}}(V(x_t^2,w_t^2))\,,{\bar\gamma _{13}}(V(x_t^3,w_t^3))\,)\\
	\,\,\,\,\,\,\,\,\,\,\,\,\,\,\,\,\,\,\,\,\,\,\,\,\,\,\,\,\,\,\,\,\,\,\,\,\,\,\,\,\,\,\,V(x_t^2,w_t^2) > {\bar\gamma _{23}}(V(x_t^3,w_t^3), V(x_	t^3,w_t^3)) > 0\\
	\,\,\,\,\,\,\,\,\,\,\,\,\,\,\,\,\,\,\,\,\,\,\,\,\,\,\,\,\,\,\,\,\,\,\,\,\,\,\,\,\,\,\, \vdots 
	\end{array}
	\end{equation} 
	This iterative procedure reduces to \eqref{E:SGC_expanded}. Hence, the team is stable irrespective of the network topology as long as it is at least weakly connected, provided it obeys certain small gain conditions.
\end{proof}
\section{Simulation Results}\label{Sec:Sim_res}
Consider a fleet of 5 autonomous vehicles moving in the horizontal plane, with the following continuous-time models (discretized at $T$=0.1s) having similar dynamics (for simplicity):
${m^i}{{\ddot x}^i} =  - \mu _1^i{{\dot x}^i} + \left( {u_R^i + u_L^i} \right)\cos {\theta ^i},
{m^i}{{\ddot y}^i} =  - \mu _1^i{{\dot y}^i} + \left( {u_R^i + u_L^i} \right)\sin {\theta ^i},
{J^i}{{\ddot \theta }^i} =  - \mu _1^i{{\dot \theta }^i} + \left( {u_R^i + u_L^i} \right){r_v}$ 
where $m$ ,$J$, $\mu_{1,2}$ and $r_v$ are parameters specified in \cite{Franco2008}). Constraints on inputs and states are $0\le|u^i_{R,L}|\le 6$ $|\dot{\theta}|\le 1\,rad/s$. Uniformly distributed communication delay is bounded by $T\le\Delta_{ij}\le 6T$. Distributed cost for each agent (leader $A^1$):
\begin{equation*}\label{sim_cost}
\begin{array}{l}
J_t^i = \sum\limits_t^{t + N_p^i - 1} {\left( {\left\| {{{\bar z}^i}_k} \right\|_{{Q^i}}^2 + \left\| {u_k^i} \right\|_{{R^i}}^2 + \sum\limits_{j \in {G^i}} {\left\| {\bar{ \bar{z}}_k^i} \right\|_{{S^{ij}}}^2} } \right)} {\kern 1pt} {\kern 1pt} \\
\,\,\,\,\,\,\,\,\,\, + \left\| {\bar z_{t + N_p^i}^i} \right\|_{Q_f^i}^2 + \sum\limits_{j \in {G^i}} {\frac{{\bar \lambda {R_{{{\min }^i}}}{{\bf{1}}_{({R_{{{\min }^j}}} - d_k^{ij}) > 0,\forall t \le k \le (t + N_p^i)}}}}{{\sum\limits_{k = t}^{t + N_P^i} {\lambda (d_k^{ij})d_k^{ij}} }}} 
\end{array}
\end{equation*}
where $\bar z^i_k=z_k^i - g^i_k + a^{i1}$ and $\bar{ \bar{z}}^i_k=z_l^i - \tilde w_l^j + a^{ij}$. Goal $g^i_k$ is the way-point (WP) for leader  and for followers it is the leader's planned trajectory, i.e. $g^i=\tilde w^1_l, \forall{i\ne{1}}$. Alignment vectors $a^{ij}$ define the formation geometry such that adjacent agents occupy positions 15 units apart in a 30 units equilateral triangle with same speed and direction. Optimization parameters for all agents are: $N_p$=50, $N_c$=15, $Q$=0.1 $diag({1,1,10,1,10,1})$, $R$=0.01 $diag({1,1})$, $S^{ij}$=0.25 $Q$ and $S^{1j}$=0.2 $S^{ij}$, for $i=1\dots 5$, $j\in{G^{i\backslash{1}}}$. For spatially filtered CA potential \eqref{E:colav_weight_condition}, parameters are $R_{\min}=5m$,$v_{\max}=40m/s$, $L^i_{hx}=\lambda_{Q,\max} |\bar z_o^i|$, $L^i_{qx}=(N-1)\lambda_{S,\max} |\bar z_o^i|$, $L^i_{hf}=\lambda_{Qf,\max} |\bar z_o^i|$ and $\underbar{r}^i=\lambda_{Q,\min} |\bar z^i_k|^2$, where $\lambda_\Pi$ is eigenvalue of $\Pi$.
 Local control $K^i$ and terminal weight $Q_f^i$ can be determined solving the LMI equation presented in  Section \ref{Sec:Stab}. Simulations were run on 3.3 GHz Intel (quad) Core i7-2500 using parallelized Matlab code, where 1 simulation second took 94 CPU seconds (which can be reduced on dedicated hardware and optimized code). For NN we use a network with 6 inputs, $H_L$=6 hidden layer neurons and 6 outputs. Thus there are 84 NN weights and biases as opposed to full trajectory of 300 floating-points (data compression of 72\%). We only show results of weakly connected network due to lack of space. $A^{4}$ and $A^5$ have only directed link from $A^{2}$ and $A^3$, making the network topology weakly connected, see inset of  Fig.\ref{F:big_traj}. Executing sharp turns, such as right angle turns when transitioning between WPs puts agents on the inside of the turn ($A^{2}$, $A^4$) at risk of collision. Also, $A^{4,5}$ receive WP information, with extra delay due to multiple hops, i.e. $\bar{\Delta}_4=\bar{\Delta}_5=2\bar{\Delta}$. However, collision is successfully avoided throughout the trajectory.  Synchronization of states is achieved quickly, as shown in Fig. \ref{F:big_states}. The effect of delay is manifest in lag in synchronization, while temporary divergence is due to collision avoidance. It is evident that the proposed algorithm performs well despite large random delays.  
 \begin{figure}
\begin{center}
\epsfig {file=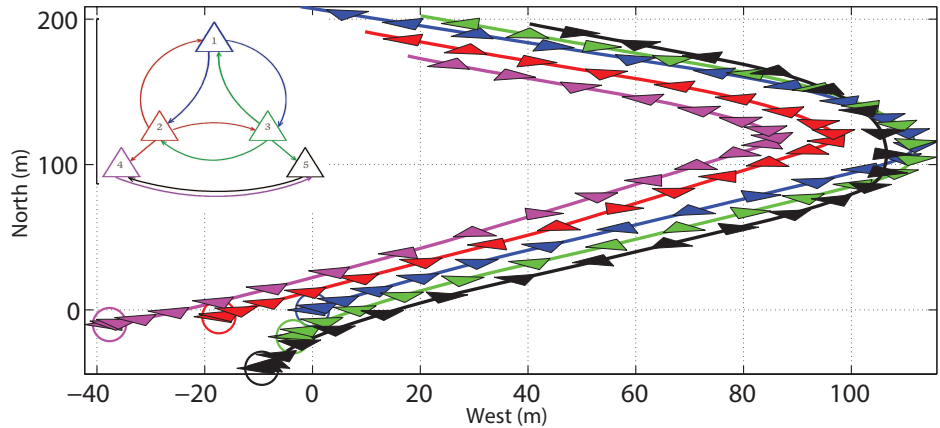, width=12.5cm}
\caption{Trajectory of agents in weakly connected network (inset: net topology).}
\label{F:big_traj}
\end{center}
\end{figure}

\begin{figure}
\begin{center}
\epsfig{file=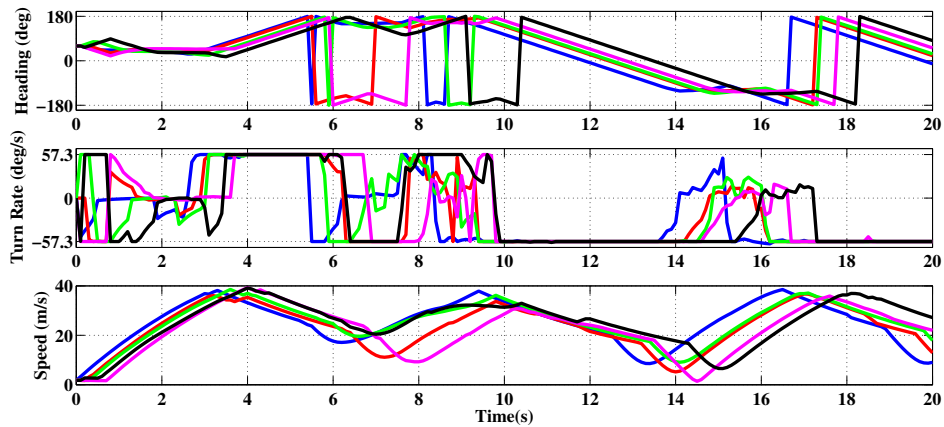,width=12.5cm}
\caption{States of agents connected in weakly connected  network.}
\label{F:big_states}
\end{center}
\end{figure}
In the given example, cost function and corresponding gain is shown in Fig. \ref{F:big_sgc1} to illustrate verification of small gain condition \eqref{E:SGC_expanded} from Theorem \ref{T:SGC} for $\bar k ̅^i=5\times 10^3$. ). The condition for only Agent 1 (connected to Agents 2 and 3 in the weakly connected network) is shown. However, small gain conditions hold for all the other agents (results not shown in interest of brevity).
\begin{figure}
	\begin{center}
		\epsfig{file=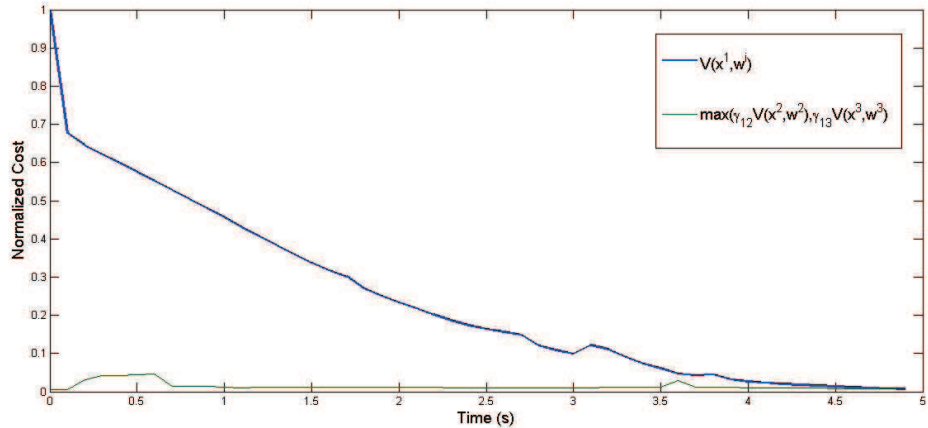,width=12.5cm}
		\caption{Small gain condition for Agent 1}
		\label{F:big_sgc1}
	\end{center}
\end{figure}
\section{Conclusion}\label{Sec:concl}
We presented distributed NMPC framework for formation control of constrained agents robust to uncertainty due to data compression and propagation delays. Collision avoidance is ensured by means of spatially filtered potential field. Rigorous proofs are provided ensuring practical stability regardless of network topology. Simulations illustrate good performance of the proposed scheme in both strongly- and weakly-connected networks. Future research directions include the need to cater for model uncertainty, disturbances and fault tolerance.
\section*{ Acknowledgment}
Support provided by King Abdulaziz City for Science \& Technology  through King Fahd University of
Petroleum \& Minerals for this work through project No. 09-SPA783-04 is acknowledged. 
\bibliographystyle{IEEEtran}        % Include this if you use bibtex 
\bibliography{journ_abbrv_ref}
\bibliographystyle{IEEEtran}        % Include this if you use bibtex 

\end{document}